\documentclass{amsart}

\usepackage[hyphens]{url}
\usepackage[bookmarks]{hyperref}
\usepackage[hyphenbreaks]{breakurl}
\usepackage[style=numeric,sorting=none,backend=biber]{biblatex}
\bibliography{gadt}
\usepackage{tiz-paper}

\usepackage[links,references]{latex/agda}
\usepackage{newunicodechar}
\newunicodechar{λ}{\ensuremath{\mathnormal\lambda}}
\newunicodechar{←}{\ensuremath{\mathnormal\from}}
\newunicodechar{→}{\ensuremath{\mathnormal\to}}
\newunicodechar{∀}{\ensuremath{\mathnormal\forall}}
\newunicodechar{≡}{\ensuremath{\mathrel{\equiv}}}
\newunicodechar{ℓ}{\ensuremath{\mathrel{\ell}}}
\newunicodechar{⊎}{\ensuremath{\mathrel{\uplus}}}
\newunicodechar{ℕ}{\ensuremath{\mathnormal{\mathbb N}}}
\newunicodechar{∷}{\ensuremath{\mathnormal{::}}}
\newunicodechar{¹}{\ensuremath{\mathnormal{^1}}}
\newunicodechar{₀}{\ensuremath{\mathnormal{_0}}}
\newunicodechar{₁}{\ensuremath{\mathnormal{_1}}}
\newunicodechar{⊤}{\ensuremath{\mathnormal{\top}}}
\newunicodechar{∞}{\ensuremath{\mathnormal{\infty}}}

\newcommand{\kwdata}{\kw{data}}
\newcommand{\vars}[1]{\textsf{vars}(#1)}
\newcommand{\toTerm}[1]{\textsf{term}(#1)}
\newcommand{\matches}[2]{\textsf{matches}(#1, #2)}
\newcommand{\clauseOk}[4]{\textsf{clause}(#1:#2, #3:#4)}
\newcommand{\ctorOk}[3]{\textsf{clause}(#1:#2, #3~~\Delta_{#3})}
\newcommand{\zerocon}{\cons{zero}}
\newcommand{\succon}{\cons{suc}}
\newcommand{\Nat}{\Dat{\mathbb N}}

\usepackage{tiz-note}
\hypersetup{colorlinks=true,linkcolor=blue}

\begin{document}

\title{A simpler encoding of indexed types}

\author{Tesla Zhang}
\address{The Pennsylvania State University}
\email{yqz5714@psu.edu}
\date\today

\maketitle
\begin{abstract}
  In functional programming languages, generalized algebraic data types (GADTs) are very useful
  as the unnecessary pattern matching over them can be ruled out by the failure of unification of type arguments.
  In dependent type systems, this is usually called \textit{indexed types} and it's particularly useful
  as the identity type is a special case of it.
  However, pattern matching over indexed types is very complicated as it requires term unification in general.
  We study a simplified version of indexed types (called \textit{simpler indexed types})
  where we explicitly specify the selection process of constructors, and we discuss its expressiveness, limitations,
  and properties.
\end{abstract}

\section{Introduction}
\label{sec:intro}
Correct-by-construction data structures are pleasant to work with,
such as well-typed and well-scoped syntax trees~\cite{AST}.
A key aspect of correct-by-construction data structures is that they carry value-level
information in their types. We start from two simple instances of such data structures:
the finite set type and the sized vector type,
which are base on the general notion of \textit{indexed types}\footnote{
This piece of code is written in Agda~\cite{Agda}.
The types and functions are in blue while constructors are in green.
There are other pieces of code written in Aya~\cite{Aya}, where types are in green
and constructors are in purple.
In~\cref{sec:tyck}, we use the latter coloring.}:

\begin{code}[hide]%
\>[0]\AgdaSymbol{\{-\#}\AgdaSpace{}%
\AgdaKeyword{OPTIONS}\AgdaSpace{}%
\AgdaPragma{--cubical}\AgdaSpace{}%
\AgdaSymbol{\#-\}}\<%
\\
\>[0]\AgdaKeyword{open}\AgdaSpace{}%
\AgdaKeyword{import}\AgdaSpace{}%
\AgdaModule{Cubical.Core.Everything}\<%
\\
\>[0]\AgdaKeyword{open}\AgdaSpace{}%
\AgdaKeyword{import}\AgdaSpace{}%
\AgdaModule{Cubical.Foundations.Prelude}\<%
\\
\>[0]\AgdaKeyword{open}\AgdaSpace{}%
\AgdaKeyword{import}\AgdaSpace{}%
\AgdaModule{Cubical.Data.Nat.Base}\<%
\\
\>[0]\AgdaKeyword{open}\AgdaSpace{}%
\AgdaKeyword{import}\AgdaSpace{}%
\AgdaModule{Agda.Primitive}\<%
\\
\>[0]\AgdaKeyword{variable}\AgdaSpace{}%
\AgdaGeneralizable{ℓ}\AgdaSpace{}%
\AgdaSymbol{:}\AgdaSpace{}%
\AgdaPostulate{Level}\<%
\end{code}

\lessSpace{-0.8}
\begin{center}
\begin{code}%
\>[0]\AgdaKeyword{data}\AgdaSpace{}%
\AgdaDatatype{Fin}\AgdaSpace{}%
\AgdaSymbol{:}\AgdaSpace{}%
\AgdaDatatype{ℕ}\AgdaSpace{}%
\AgdaSymbol{→}\AgdaSpace{}%
\AgdaPrimitive{Type₀}\AgdaSpace{}%
\AgdaKeyword{where}\<%
\\
\>[0][@{}l@{\AgdaIndent{0}}]%
\>[2]\AgdaInductiveConstructor{fzero}%
\>[9]\AgdaSymbol{:}\AgdaSpace{}%
\AgdaSymbol{∀}\AgdaSpace{}%
\AgdaSymbol{\{}\AgdaBound{n}\AgdaSymbol{\}}\AgdaSpace{}%
\AgdaSymbol{→}\AgdaSpace{}%
\AgdaDatatype{Fin}\AgdaSpace{}%
\AgdaSymbol{(}\AgdaInductiveConstructor{suc}\AgdaSpace{}%
\AgdaBound{n}\AgdaSymbol{)}\<%
\\
\>[2]\AgdaInductiveConstructor{fsuc}%
\>[9]\AgdaSymbol{:}\AgdaSpace{}%
\AgdaSymbol{∀}\AgdaSpace{}%
\AgdaSymbol{\{}\AgdaBound{n}\AgdaSymbol{\}}\AgdaSpace{}%
\AgdaSymbol{→}\AgdaSpace{}%
\AgdaDatatype{Fin}\AgdaSpace{}%
\AgdaBound{n}\AgdaSpace{}%
\AgdaSymbol{→}\AgdaSpace{}%
\AgdaDatatype{Fin}\AgdaSpace{}%
\AgdaSymbol{(}\AgdaInductiveConstructor{suc}\AgdaSpace{}%
\AgdaBound{n}\AgdaSymbol{)}\<%
\\
\\[\AgdaEmptyExtraSkip]%
\>[0]\AgdaKeyword{data}\AgdaSpace{}%
\AgdaDatatype{Vect}\AgdaSpace{}%
\AgdaSymbol{(}\AgdaBound{A}\AgdaSpace{}%
\AgdaSymbol{:}\AgdaSpace{}%
\AgdaPrimitive{Type}\AgdaSpace{}%
\AgdaGeneralizable{ℓ}\AgdaSymbol{)}\AgdaSpace{}%
\AgdaSymbol{:}\AgdaSpace{}%
\AgdaDatatype{ℕ}\AgdaSpace{}%
\AgdaSymbol{→}\AgdaSpace{}%
\AgdaPrimitive{Type}\AgdaSpace{}%
\AgdaBound{ℓ}\AgdaSpace{}%
\AgdaKeyword{where}\<%
\\
\>[0][@{}l@{\AgdaIndent{0}}]%
\>[2]\AgdaInductiveConstructor{[]}%
\>[9]\AgdaSymbol{:}\AgdaSpace{}%
\AgdaDatatype{Vect}\AgdaSpace{}%
\AgdaBound{A}\AgdaSpace{}%
\AgdaNumber{0}\<%
\\
\>[2]\AgdaOperator{\AgdaInductiveConstructor{\AgdaUnderscore{}∷\AgdaUnderscore{}}}%
\>[8]\AgdaSymbol{:}\AgdaSpace{}%
\AgdaSymbol{∀}\AgdaSpace{}%
\AgdaSymbol{\{}\AgdaBound{n}\AgdaSymbol{\}}\AgdaSpace{}%
\AgdaSymbol{→}\AgdaSpace{}%
\AgdaBound{A}\AgdaSpace{}%
\AgdaSymbol{→}\AgdaSpace{}%
\AgdaDatatype{Vect}\AgdaSpace{}%
\AgdaBound{A}\AgdaSpace{}%
\AgdaBound{n}\AgdaSpace{}%
\AgdaSymbol{→}\AgdaSpace{}%
\AgdaDatatype{Vect}\AgdaSpace{}%
\AgdaBound{A}\AgdaSpace{}%
\AgdaSymbol{(}\AgdaInductiveConstructor{suc}\AgdaSpace{}%
\AgdaBound{n}\AgdaSymbol{)}\<%
\end{code}
\end{center}
\lessSpace{-0.8}

The indices of types are the parameters at the right-hand-side of the colons in the signatures
of inductive types, which can be specialized by constructors.
The two constructors of \AgdaDatatype{Fin} specify the \textit{index} as $\AgdaInductiveConstructor{suc}~n$,
so when pattern matching over $\AgdaDatatype{Fin}~\AgdaInductiveConstructor{zero}$ requires no clauses.
The algorithm for selecting constructors is a process of term unification, extracting a most-general-unifier
and apply that to the rest of the telescope.
For each constructor, we unify the type arguments with the indices it specifies,
and there are three potential results~\cite[\S 2.1]{NoK}.

\begin{itemize}
\item Success positively -- the constructor matches.
\item Success negatively -- the constructor does not matches.
\item Failure -- cannot decide, pattern matching cannot be performed.
\end{itemize}

Users will have to understand the error messages with unification failures,
which is an accidental complexity brought into dependent type systems.
The implementation of the unification algorithm also affects the selection of constructors.

We propose an alternative syntax for indexed types.
First, for inductive types without indices, we use a Haskell-style syntax to describe its arguments,
and we allow bindings in the parameters since we are working with dependent types:

\lessSpace{-1.2}
\begin{figure}[h]
\centering
\subfloat{
$\begin{aligned}
&\kwdata~\Nat:\UU \\[-0.3em]
& \mid~\zerocon \\[-0.3em]
& \mid~\succon~(x:\Nat)
\end{aligned}$}
\qquad
\subfloat{
$\begin{aligned}
&\kwdata~\Dat{List}~(A:\UU):\UU \\[-0.3em]
& \mid~\cons{nil} \\[-0.3em]
& \mid~\cons{cons}~(x:A)~(xs:\Dat{List}~A)
\end{aligned}$}
\end{figure}

Then, we allow the constructors to perform a pattern matching over the type of the parameters.
For instance, we define the sized vector type using the following syntax:

\lessSpace{-1.2}
\begin{align*}
&\kwdata~\Dat{Vec}~(A:\UU)~(n:\Nat):\UU \\[-0.3em]
& \mid~A,\zerocon~\Rightarrow\cons{vnil} \\[-0.3em]
& \mid~A,\succon~n~\Rightarrow\cons{vcons}~(x:A)~(xs:\Dat{Vec}~A~n)
\end{align*}

This pattern matching is not a traditional \textit{pattern matching}, say,
it does not need to be covering (although in the $\Dat{Vec}$ example it is)
and it can contain seemingly unreachable patterns (like duplicated patterns).
Instead, they represent the selection process of constructors directly.
The type checking of pattern matching consists of two steps: the well-typedness of
patterns and the exhaustiveness of the patterns.
We exemplify the type checking of our encoding of indexed types by describing
the pattern matching over $\Dat{Vec}~\Nat~n$.
First, it tries to match the terms $\Nat, n$ with the patterns $A, \zerocon$ and $A, \succon~n$.
The pattern matching has three potential results, similar to the term unification problem:

\begin{itemize}
\item Success positively -- the patterns are matched,
 this constructor will be available (needs to be matched).
\item Success negatively -- the patterns do not match,
 this constructor is not available (does not need to be matched).
\item Failure -- the pattern matching gets stuck, pattern matching cannot be performed.
\end{itemize}

However, pattern matching is a basic construct in dependent type systems,
and it is decidable and terminating -- unlike the general term unification problem,
where we normally give up higher-order cases to avoid undecidability.
It is also more friendly to general users because they are required to understand one concept less.

Another example is the finite set type:

\lessSpace{-1.2}
\begin{align*}
&\kwdata~\Dat{Fin}~(n:\Nat):\UU \\[-0.3em]
& \mid~\succon~n~\Rightarrow\cons{fzero} \\[-0.3em]
& \mid~\succon~n~\Rightarrow\cons{fsuc}~(x:\Dat{Fin}~n)
\end{align*}

\subsection{Problematic indexed types} 
\label{sub:motive}

We propose the new syntax because the term unification problem
generated by general indexed types could be very complicated.
Here are some examples of indexed types in Agda that generate such unification problems:

\begin{code}[hide]%
\>[0]\AgdaSymbol{\{-\#}\AgdaSpace{}%
\AgdaKeyword{OPTIONS}\AgdaSpace{}%
\AgdaPragma{--cubical}\AgdaSpace{}%
\AgdaPragma{--with-K}\AgdaSpace{}%
\AgdaSymbol{\#-\}}\<%
\\
\>[0]\AgdaKeyword{open}\AgdaSpace{}%
\AgdaKeyword{import}\AgdaSpace{}%
\AgdaModule{Cubical.Core.Everything}\<%
\\
\>[0]\AgdaKeyword{open}\AgdaSpace{}%
\AgdaKeyword{import}\AgdaSpace{}%
\AgdaModule{Cubical.Data.Nat.Base}\AgdaSpace{}%
\AgdaKeyword{renaming}\AgdaSpace{}%
\AgdaSymbol{(}\AgdaFunction{predℕ}\AgdaSpace{}%
\AgdaSymbol{to}\AgdaSpace{}%
\AgdaFunction{pred}\AgdaSymbol{;}\AgdaSpace{}%
\AgdaOperator{\AgdaPrimitive{\AgdaUnderscore{}∸\AgdaUnderscore{}}}\AgdaSpace{}%
\AgdaSymbol{to}\AgdaSpace{}%
\AgdaOperator{\AgdaPrimitive{\AgdaUnderscore{}-\AgdaUnderscore{}}}\AgdaSymbol{)}\<%
\\
\>[0]\AgdaKeyword{open}\AgdaSpace{}%
\AgdaKeyword{import}\AgdaSpace{}%
\AgdaModule{Cubical.Data.Bool.Base}\<%
\\
\>[0]\AgdaKeyword{open}\AgdaSpace{}%
\AgdaKeyword{import}\AgdaSpace{}%
\AgdaModule{Agda.Primitive}\<%
\\
\>[0]\AgdaKeyword{variable}\AgdaSpace{}%
\AgdaGeneralizable{A}\AgdaSpace{}%
\AgdaGeneralizable{B}\AgdaSpace{}%
\AgdaSymbol{:}\AgdaSpace{}%
\AgdaPrimitive{Set}\<%
\\
\>[0]\AgdaKeyword{variable}\AgdaSpace{}%
\AgdaGeneralizable{ℓ}\AgdaSpace{}%
\AgdaSymbol{:}\AgdaSpace{}%
\AgdaPostulate{Level}\<%
\end{code}

\lessSpace{-0.8}
\begin{center}
\begin{code}%
\>[0]\AgdaKeyword{data}\AgdaSpace{}%
\AgdaDatatype{Univ}\AgdaSpace{}%
\AgdaSymbol{:}\AgdaSpace{}%
\AgdaPrimitive{Type₀}\AgdaSpace{}%
\AgdaSymbol{→}\AgdaSpace{}%
\AgdaPrimitive{Type₁}\AgdaSpace{}%
\AgdaKeyword{where}\<%
\\
\>[0][@{}l@{\AgdaIndent{0}}]%
\>[2]\AgdaInductiveConstructor{univ}\AgdaSpace{}%
\AgdaSymbol{:}\AgdaSpace{}%
\AgdaSymbol{∀}\AgdaSpace{}%
\AgdaBound{u}\AgdaSpace{}%
\AgdaSymbol{→}\AgdaSpace{}%
\AgdaDatatype{Univ}\AgdaSpace{}%
\AgdaSymbol{(}\AgdaBound{u}\AgdaSpace{}%
\AgdaSymbol{→}\AgdaSpace{}%
\AgdaDatatype{ℕ}\AgdaSymbol{)}\<%
\\
\>[0]\AgdaKeyword{data}\AgdaSpace{}%
\AgdaDatatype{Higher}\AgdaSpace{}%
\AgdaSymbol{:}\AgdaSpace{}%
\AgdaSymbol{(}\AgdaDatatype{ℕ}\AgdaSpace{}%
\AgdaSymbol{→}\AgdaSpace{}%
\AgdaDatatype{ℕ}\AgdaSymbol{)}\AgdaSpace{}%
\AgdaSymbol{→}\AgdaSpace{}%
\AgdaPrimitive{Type₀}\AgdaSpace{}%
\AgdaKeyword{where}\<%
\\
\>[0][@{}l@{\AgdaIndent{0}}]%
\>[2]\AgdaInductiveConstructor{higher-suc}%
\>[15]\AgdaSymbol{:}\AgdaSpace{}%
\AgdaDatatype{Higher}\AgdaSpace{}%
\AgdaInductiveConstructor{suc}\<%
\\
\>[2]\AgdaInductiveConstructor{higher-pred}%
\>[15]\AgdaSymbol{:}\AgdaSpace{}%
\AgdaDatatype{Higher}\AgdaSpace{}%
\AgdaFunction{pred}\<%
\\
\>[2]\AgdaInductiveConstructor{higher-misc}%
\>[15]\AgdaSymbol{:}\AgdaSpace{}%
\AgdaDatatype{Higher}\AgdaSpace{}%
\AgdaSymbol{(λ}\AgdaSpace{}%
\AgdaBound{x}\AgdaSpace{}%
\AgdaSymbol{→}\AgdaSpace{}%
\AgdaNumber{2}\AgdaSpace{}%
\AgdaOperator{\AgdaPrimitive{+}}\AgdaSpace{}%
\AgdaBound{x}\AgdaSpace{}%
\AgdaOperator{\AgdaPrimitive{-}}\AgdaSpace{}%
\AgdaNumber{3}\AgdaSymbol{)}\<%
\end{code}
\end{center}
\lessSpace{-0.8}

We cannot perform pattern matching on the type $\AgdaDatatype{Univ Bool}$
since Agda cannot unify \AgdaDatatype{Bool} and \AgdaDatatype{ℕ}
(this unification problem is related to the injectivity of type constructors, see the discussion in~\cite[\S 1]{Unifiers}).
Similarly the unification problem may become higher order (like in \AgdaDatatype{Higher} --
neither $\AgdaDatatype{Higher}~\AgdaInductiveConstructor{suc}$, $\AgdaDatatype{Higher~pred}$,
$\AgdaDatatype{Higher}~(\lambda~x\to\AgdaNumber{2}~\AgdaFunction{+}~x~\AgdaFunction{-}~\AgdaNumber{3})$
could be pattern matched against!), generating more confusing instances.

These examples are impossible to construct with simpler indexed types.
With the proposed syntax,
we could avoid not only implementing such a complicated unification algorithm,
but also explaining these unification failures in error messages.


\subsection{Contributions}
\begin{itemize}
\item We present the syntax (\cref{sub:core}) and the type checking algorithm
 (\cref{sub:typing}) for simpler indexed types in~\cref{sec:tyck}.
\item We discuss their limitations (\cref{sub:lim}), provide a translation of
 simpler indexed types to general indexed types (\cref{sub:conserv}),
 and discuss the compilation of simpler indexed types (\cref{sub:erasure}) in~\cref{sec:meta}.
\item We explore potential extension to simpler indexed types (\cref{sub:future})
 and compare it with similar work (\cref{sub:related}) in~\cref{sec:concl}.
\end{itemize}

\section{Formalization of simpler indexed types}
\label{sec:tyck}

In this section, we describe the core language syntax and the type checking of simpler indexed types.
The coverage checking can be adapted from any other dependent type systems with indexed types
by replacing the term unification with pattern matching, so we assume the existence of a suitable coverage check.

\subsection{Core language syntax} 
\label{sub:core}

The syntax of terms is presented in~\cref{fig:syntax}.
It has spine-normal fully-applied applications on ``definitions''
(including types $\Dat D$, constructors $\cons c$, and functions definitions $\func f$).
Normal constructs such as $\lambda$-abstraction and the $\Pi$-type are also available.
$\Rightarrow$ is used in $\lambda$-abstractions instead of dots for
consistency with function definitions and pattern matching clauses.
We use $\overline u$ to denote a list of expressions, and $\emptyset$ when the list is empty.

\begin{figure}[h!]
\begin{align*}
  x,y     ::= & && \text{variable names} \\[-0.3em]
  A,B,u,v ::= & \quad \func f~\overline u && \text{full application to functions} \\[-0.3em]
          \mid & \quad x~\overline u && \text{application to references} \\[-0.3em]
          \mid & \quad \Dat D~\overline u && \text{fully applied inductive type} \\[-0.3em]
          \mid & \quad \cons c~\overline u && \text{fully applied constructor} \\[-0.3em]
          \mid & \quad (x:A)\to B && \text{$\Pi$-type} \\[-0.3em]
          \mid & \quad \lam x u && \text{lambda abstraction}
\end{align*}
\caption{Syntax of terms}
\label{fig:syntax}
\end{figure}

Case-split expressions can be encoded as functions and can be easily added
to our type theory, but they are unrelated to simpler indexed types. Therefore, we omit them.
We will have two syntactic sugars for the $\Pi$-type: $A \to B$ for $(x:A) \to B$,
and $\Delta \to B$ for $(x_1:A_1) \to (x_2:A_2)\to\cdots\to(x_n:A_n)\to B$
where $\Delta=(\overline{x_i:A_i})_{\forall i\in[1,n]}$.
The latter is only used in~\cref{sub:conserv}.

The syntax for definitions, contexts and signatures is defined in~\cref{fig:sig}.
A signature is a list of declarations and a context is a list of bindings.
Constructors are with or without a list of patterns.
The variables in the same pattern are assumed to be distinct.

\begin{figure}[h!]
\begin{align*}
  \Gamma,\Delta,\Theta ::= & \quad \overline{x_i:A_i} && \text{context} \\[-0.3em]
  decl ::= & \quad \kwdata~\Dat D~\Delta~\overline{cons} && \text{simpler indexed type} \\[-0.3em]
      \mid & \quad \kw{func}~\func f~\Delta:A~\overline{cls} && \text{function definition} \\[-0.3em]
  cons ::= & \quad \mid \overline p\Rightarrow \cons c~\Delta && \text{pattern matching constructor} \\[-0.3em]
      \mid & \quad \mid \cons c~\Delta && \text{constructor} \\[-0.3em]
  cls  ::= & \quad \mid \overline p\Rightarrow u && \text{pattern matching clause} \\[-0.3em]
  p, q ::= & \quad \cons c~\overline p && \text{constructor patterns} \\[-0.3em]
      \mid & \quad x && \text{catch-all patterns} \\[-0.3em]
      \mid & \quad \kw{impossible} && \text{absurd patterns} \\[-0.3em]
  \Sigma ::= & \quad \overline{decl} && \text{signature}
\end{align*}
\caption{Syntax of signature and declarations}
\label{fig:sig}
\end{figure}

We will borrow some notational convention from~\cite[\S 3.2]{DepPM}\footnote{
Other styles of substitution include $u[x\mapsto v]$, $u[x/v]$, etc. (there is a relevant online discussion
\cite{SubstNotation})}:
\fbox{$u[v/x]$} for substituting occurrences of $x$ with $v$ in term $u$.
We use \fbox{$u[\overline v/\overline x]$} to denote a list of substitutions applied sequentially
to the term $u$. Substitution objects are denoted as $\sigma$.
We will assume the substitution operation defined on terms, patterns, and substitutions.

In the typing rules in~\cref{sub:typing}, we will omit the vertical bars in $cons$ and $cls$
which are intended to separate the clauses and constructors.


\subsection{Operations on terms} 
\label{sub:ops}

We also need some operations on terms and patterns.
All of them are defined by induction on the syntax.
We define $\vars\Delta$ to compute the list of variables in $\Delta$:

\lessSpace{-1.2}
\begin{align*}
  \vars\emptyset &:= \emptyset \\
  \vars{x:A,\Delta} &:= x,\vars\Delta
\end{align*}

We define {$\vars p$} to compute the list of bindings in pattern $p$
and {$\vars{\overline p}$} to gather all the bindings in the patterns $\overline p$.
This operation requires the well-typedness of the patterns because we need the types of the bindings.
We store these types into the patterns to allow accessing them in this operation:

\lessSpace{-1.2}
\begin{align*}
  \vars{x:A} &:= x:A \\
  \vars{\kw{impossible}} &:= \emptyset \\
  \vars{\cons c~\overline p} &:= \vars{\overline p} \\
  \vars{\emptyset} &:= \emptyset \\
  \vars{x:A,\overline p} &:= x:A,\vars{\overline p}
\end{align*}

We define {$\toTerm p$} to compute a term that matches exactly the pattern $p$
and $\toTerm{\overline p}$ to compute a list of terms matching exactly the patterns $\overline p$.
This requires $p$ to contain no $\kw{impossible}$ sub-patterns:

\lessSpace{-1.2}
\begin{align*}
  \toTerm{x:A} &:= x \\
  \toTerm{\cons c~\overline p} &:= \cons c~\toTerm{\overline p} \\
  \toTerm{\emptyset} &:= \emptyset \\
  \toTerm{q,\overline p} &:= \toTerm q,\toTerm{\overline p}
\end{align*}

We define $\matches u p\mapsto \sigma$ to perform pattern matching,
similar to the \textsc{Match} and \textsc{Matches} operations in~\cite{Goguen06}.
It computes a substitution when the pattern matching success positively
and produces $\bot$ when the pattern matching success negatively.
We will also define a version of this operation to match a list of terms with a list
of patterns $\matches{\overline u}{\overline p} \mapsto \sigma$,
similar to $\vars{\overline p}$ and $\toTerm{\overline p}$.

\begin{figure}[ht]
\begin{mathpar}
\inferrule{}{\matches u x \mapsto [u/x]} \and
\inferrule{}{\matches \emptyset \emptyset \mapsto []} \and
\inferrule{\matches{\overline u}{\overline p} \mapsto \sigma}
{\matches{\cons c~\overline u}{\cons c~\overline p} \mapsto \sigma} \and
\inferrule{\matches{\overline u}{\overline p} \mapsto \bot}
{\matches{\cons c~\overline u}{\cons c~\overline p} \mapsto \bot} \and
\inferrule{\cons{c_1} \neq \cons{c_2}}
{\matches{\cons {c_1}~\overline u}{\cons {c_2}~\overline p} \mapsto \bot} \and
\inferrule{\matches{\overline u}{\overline p} \mapsto \bot}
{\matches{(v,\overline u)}{(q, \overline p)} \mapsto \bot} \and
\inferrule{\matches v q \mapsto \bot}
{\matches{(v,\overline u)}{(q, \overline p)} \mapsto \bot} \and
\inferrule{\matches{\overline u}{\overline p} \mapsto \sigma \\
\matches v q \mapsto \sigma'}
{\matches{(v,\overline u)}{(q, \overline p)} \mapsto \sigma \uplus \sigma'}
\end{mathpar}
\caption{Pattern matching operation}
\end{figure}

\begin{lem}
For all pattern $p$, $\matches{\toTerm p} p \mapsto \sigma$ and for all list of patterns
$\overline p$, $\matches{\toTerm{\overline p}}{\overline p} \mapsto \sigma$.
In both formulae, the substitution $\sigma$ is an identity substitution.
\end{lem}
\begin{proof}
By induction on $p$.
\end{proof}


\subsection{Typing rules for terms} 
\label{sub:typing}

Well-typed terms are formed under the following type checking judgments:

\begin{itemize}
\item \fbox{$\SGvdash \Delta$} $\Delta$ is a well-formed context under $\Sigma;\Gamma$.
\item \fbox{$\SGvdash u:A$} term $u$ has type $A$ under $\Sigma;\Gamma$.
\item \fbox{$\SGvdash \overline u:\Delta$} terms $\overline u$ instantiate context $\Delta$ under $\Sigma;\Gamma$.
\item \fbox{$\SGvdash u=v:A$} terms $u$ and $v$ are equal inhabitants of type $A$ under $\Sigma;\Gamma$.
\end{itemize}

Typing rules for types and terms are defined in~\cref{fig:terms}.
They are grouped by the relevant type formation.

For simplicity, we will omit several things:
\begin{itemize}
\item  We assume the conversion check between terms --
the problem is the same as other dependent type systems,
and the strategies used by other systems will apply to ours as well.
\item We also have the type-in-type rule to simplify the universe types,
and in practical implementations, we could integrate polymorphic universe levels to
make the system consistent.
\item In the implementation of Aya and Arend, we also have the sigma type and records,
but we omit them here for simplicity.
\end{itemize}

\begin{figure}[ht]
\RaggedRight
Rules related to the $\Pi$-type.
\begin{mathpar}
\inferrule{\SGvdash A:\UU \\ \Sigma;\Gamma,x:A\vdash B:\UU}
{\SGvdash (x:A)\to B:\UU} \and
\inferrule{\kw{func}~\func f~\Delta:A~\overline{cls}\in\Sigma \\
\SGvdash \overline v:\Delta}
{\SGvdash \func f~\overline v:A[\overline v/\vars\Delta]} \and
\inferrule{\Sigma;\Gamma,x:A\vdash b:B[x/y]}{\SGvdash \lam x b : (y:A)\to B} \and
\inferrule{\SGvdash u : (x:A)\to B \\ \SGvdash v:A}{\SGvdash u~v:B[v/x]}
\end{mathpar} \\
Rules related to indexed types.
\begin{mathpar}
\inferrule{\kwdata~\Dat D~\Delta~\overline{cons}\in\Sigma
\\ \SGvdash \overline u : \Delta}
{\SGvdash \Dat D~\overline u:\UU} \and
\inferrule{\kwdata~\Dat D~\Delta~\overline{cons} \in\Sigma \\
\SGvdash \overline u:\Delta \\\\
\cons c~\Delta_{\cons c}\in \overline{cons} \\
\SGvdash \overline v:\Delta_{\cons c}[\overline u/\vars\Delta]}
{\SGvdash \cons c~\overline v:\Dat D~\overline u}\textsc{ConCall} \and
\inferrule{\kwdata~\Dat D~\Delta~\overline{cons} \in\Sigma \\
\SGvdash \overline u:\Delta \\\\
\overline p \Rightarrow\cons c~\Delta_{\cons c}\in \overline{cons} \\
\matches u p \mapsto \sigma \\\\
\SGvdash \overline v:\Delta_{\cons c}\sigma[\overline u/\vars\Delta]}
{\SGvdash \cons c~\overline v:\Dat D~\overline u}\textsc{IxCall}
\end{mathpar} \\
Rule for convertible types and type-in-type.
\begin{mathpar}
\inferrule{\SGvdash a:A \\ \SGvdash A=B : \UU}{\SGvdash a:B}
\and \inferrule{}{\SGvdash \UU : \UU}
\end{mathpar}
\caption{Typing rules for types and terms}
\label{fig:terms}
\end{figure}

In the \textsc{IxCall} rule, we perform a pattern matching between the type arguments
and the patterns in the constructor to make sure the availability of the selected constructor,
and apply the resulting substitution to the parameters of the constructor as they can access
the patterns according to the rules in~\cref{fig:patty}.
In contrast, the \textsc{ConCall} rule does not perform pattern matching and the constructor
is directly available. The differences between \textsc{IxCall} and \textsc{ConCall} include
a successful $\matches u p$ operation and an extra substitution applied on $\Delta_{\cons c}$.


\subsection{Signature well-formedness} 
\label{sub:decl}

A well-formed signature consists of a list of well-typed declarations.
We can think of the whole type checking algorithm as a signature formation process.

To check function definitions and pattern matching constructors,
we first need to type check the patterns and elaborate the pattern matching clauses.
The rules of pattern type checking, similar to the operations in~\cref{sub:ops},
have two versions:

\begin{itemize}
\item \fbox{$\SGvdash p : A \mapsto \Theta$} type-checking a pattern $p$ against a type $A$.
\item \fbox{$\SGvdash \overline p : \Delta \mapsto \Theta$} type-checking patterns
$\overline p$ against a context $\Delta$.
\end{itemize}

These rules are defined in~\cref{fig:patty}.
They produce a context $\Theta$ containing all of the bindings in the give pattern(s).

\begin{figure}[ht]
\RaggedRight
Rules for one pattern.
\begin{mathpar}
\inferrule{}{\SGvdash x : A \mapsto x:A} \and
\inferrule{\kwdata~\Dat D~\Delta~\overline{cons}\in\Sigma \\
\cons c~\Delta_{\cons c} \in \overline{cons} \\
\SGvdash \overline p : \Delta_{\cons c} \mapsto \Theta}
{\SGvdash \cons c~\overline p : \Dat D~\overline u \mapsto \Theta} \and
\inferrule{\kwdata~\Dat D~\Delta~\overline{cons}\in\Sigma \\
\overline q \Rightarrow \cons c~\Delta_{\cons c} \in \overline{cons} \\
\matches{\overline u}{\overline q} \mapsto \sigma \\
\SGvdash \overline p : \Delta_{\cons c}\sigma \mapsto \Theta}
{\SGvdash \cons c~\overline p : \Dat D~\overline u \mapsto \Theta} \and
\inferrule{\kwdata~\Dat D~\Delta~\overline{cons}\in\Sigma \\
\left(\matches{\overline u}{\overline p} \mapsto \bot\right)_{
  \forall \overline p \Rightarrow \cons c~\Delta_{\cons c} \in \overline{cons}}\\
\cons c~\Delta_{\cons c} \notin \overline{cons}}
{\SGvdash \kw{impossible} : \Dat D~\overline u \mapsto \Theta}
\end{mathpar} \\
Rules for a list of patterns.
\begin{mathpar}
\inferrule{\SGvdash q : A \mapsto \Theta \\
\SGvdash \overline p : \Delta[\toTerm q/x] \mapsto \Theta'}
{\SGvdash q,\overline p : (x:A,\Delta) \mapsto \Theta\uplus\Theta'} \and
\inferrule{}{\SGvdash \emptyset : \emptyset \mapsto \emptyset}
\end{mathpar}
\caption{Type checking of patterns}
\label{fig:patty}
\end{figure}

\begin{lem}
\label{lem:typed-pats}
$\SGvdash \overline p : \Delta \mapsto \Theta \implies \SGvdash \toTerm{\overline p} : \Delta$
and  $\SGvdash p : A \mapsto \Theta \implies \SGvdash \toTerm p : A$.
\end{lem}

\begin{proof}
By induction on $p$.
\end{proof}

Then, we define the rules for type checking pattern matching structures as in~\cref{fig:matchy}
using the operation defined in~\cref{fig:patty}.

\begin{figure}
\begin{mathpar}
\inferrule{\SGvdash \overline p : \Delta \mapsto \Theta \\\\
\Sigma;\Gamma,\Delta,\Theta \vdash u : A[\toTerm{\overline p}/\vars\Delta]}
{\SGvdash \clauseOk{\overline p}\Delta u A} \and
\inferrule{\SGvdash \overline p : \Delta \mapsto \Theta \\
\Sigma;\Gamma,\Delta,\Theta \vdash \Delta_{\cons c}}
{\SGvdash \ctorOk{\overline p}\Delta{\cons c}}
\end{mathpar}
\caption{Pattern matching structures}
\label{fig:matchy}
\end{figure}

With them, we could define the type checking of function definitions and simpler indexed types,
and form signature by the rules \fbox{$\SGvdash$} in~\cref{fig:sign-well}.
The declarations are checked one after another so latter functions can depend on former ones.
By that, we will lose induction-recursion~\cite{IndRec} and induction-induction~\cite{IndInd},
and we consider it a potential future work.

\begin{figure}
\begin{mathpar}
\inferrule{\SGvdash \\ \SGvdash \Delta \\
\left(\Sigma;\Gamma,\Delta \vdash \Delta_{\cons c}
\right)_{\forall (\cons c~\Delta_{\cons c}) \in \overline{cons}}
\\\\\left(\inferrule{}{\Sigma;\Gamma,\Delta \vdash \Delta_{\cons c}
\\\ctorOk{\overline p}\Delta{\cons c}}
\right)_{\forall (\overline p \Rightarrow\cons c~\Delta_{\cons c}) \in \overline{cons}}}
{\Sigma,\kwdata~\Dat D~\Delta~\overline{cons}; \Gvdash}\and
\inferrule{\SGvdash \\ \SGvdash \Delta \\
\left(\clauseOk{\overline p}\Delta u A\right)_{\forall\overline p\Rightarrow u\in\overline{cls}}}
{\Sigma,\kw{func}~\func f~\Delta~:A~\overline{cls}; \Gvdash}
\end{mathpar}
\caption{Well-formedness of signature $\Sigma$}
\label{fig:sign-well}
\end{figure}


\section{Metatheory of simpler indexed types}
\label{sec:meta}

In this section, we discuss the limitations of simpler indexed types by examples and
provide a complete and sound translation from simpler indexed types to general indexed types.

\subsection{Limitations and workarounds} 
\label{sub:lim}

Many useful indexed types cannot be written as simpler indexed types.
For instance, as a common illustration of the convenience brought by indexed types,
the normalizer and the syntax tree for an expression language with
natural numbers and booleans can be defined as a general indexed type in Agda~\cite{Agda}.

\begin{code}[hide]%
\>[0]\AgdaSymbol{\{-\#}\AgdaSpace{}%
\AgdaKeyword{OPTIONS}\AgdaSpace{}%
\AgdaPragma{--cubical}\AgdaSpace{}%
\AgdaPragma{--with-K}\AgdaSpace{}%
\AgdaSymbol{\#-\}}\<%
\\
\>[0]\AgdaKeyword{open}\AgdaSpace{}%
\AgdaKeyword{import}\AgdaSpace{}%
\AgdaModule{Cubical.Core.Everything}\<%
\\
\>[0]\AgdaKeyword{open}\AgdaSpace{}%
\AgdaKeyword{import}\AgdaSpace{}%
\AgdaModule{Cubical.Data.Nat.Base}\<%
\\
\>[0]\AgdaKeyword{open}\AgdaSpace{}%
\AgdaKeyword{import}\AgdaSpace{}%
\AgdaModule{Cubical.Data.Bool.Base}\<%
\\
\>[0]\AgdaKeyword{open}\AgdaSpace{}%
\AgdaKeyword{import}\AgdaSpace{}%
\AgdaModule{Agda.Primitive}\<%
\\
\>[0]\AgdaKeyword{variable}\AgdaSpace{}%
\AgdaGeneralizable{A}\AgdaSpace{}%
\AgdaGeneralizable{B}\AgdaSpace{}%
\AgdaSymbol{:}\AgdaSpace{}%
\AgdaPrimitive{Set}\<%
\\
\>[0]\AgdaKeyword{variable}\AgdaSpace{}%
\AgdaGeneralizable{ℓ}\AgdaSpace{}%
\AgdaSymbol{:}\AgdaSpace{}%
\AgdaPostulate{Level}\<%
\end{code}

The syntax looks like this.
We first define the type \AgdaDatatype{Term} as a type indexed by another type,
which is the type of the evaluation result of each syntax variant.
In each constructor, we specialize this type.

\lessSpace{-0.8}
\begin{center}
\begin{code}%
\>[0]\AgdaKeyword{data}\AgdaSpace{}%
\AgdaDatatype{Term}\AgdaSpace{}%
\AgdaSymbol{:}\AgdaSpace{}%
\AgdaPrimitive{Type₀}\AgdaSpace{}%
\AgdaSymbol{→}\AgdaSpace{}%
\AgdaPrimitive{Type₀}\AgdaSpace{}%
\AgdaKeyword{where}\<%
\\
\>[0][@{}l@{\AgdaIndent{0}}]%
\>[2]\AgdaInductiveConstructor{nat}%
\>[8]\AgdaSymbol{:}\AgdaSpace{}%
\AgdaDatatype{ℕ}\AgdaSpace{}%
\AgdaSymbol{→}\AgdaSpace{}%
\AgdaDatatype{Term}\AgdaSpace{}%
\AgdaDatatype{ℕ}\<%
\\
\>[2]\AgdaInductiveConstructor{succ}%
\>[8]\AgdaSymbol{:}\AgdaSpace{}%
\AgdaDatatype{Term}\AgdaSpace{}%
\AgdaDatatype{ℕ}\AgdaSpace{}%
\AgdaSymbol{→}\AgdaSpace{}%
\AgdaDatatype{Term}\AgdaSpace{}%
\AgdaDatatype{ℕ}\<%
\\
\>[2]\AgdaInductiveConstructor{bool}%
\>[8]\AgdaSymbol{:}\AgdaSpace{}%
\AgdaDatatype{Bool}\AgdaSpace{}%
\AgdaSymbol{→}\AgdaSpace{}%
\AgdaDatatype{Term}\AgdaSpace{}%
\AgdaDatatype{Bool}\<%
\\
\>[2]\AgdaInductiveConstructor{inv}%
\>[8]\AgdaSymbol{:}\AgdaSpace{}%
\AgdaDatatype{Term}\AgdaSpace{}%
\AgdaDatatype{Bool}\AgdaSpace{}%
\AgdaSymbol{→}\AgdaSpace{}%
\AgdaDatatype{Term}\AgdaSpace{}%
\AgdaDatatype{Bool}\<%
\\
\>[2]\AgdaInductiveConstructor{case}%
\>[8]\AgdaSymbol{:}\AgdaSpace{}%
\AgdaDatatype{Term}\AgdaSpace{}%
\AgdaDatatype{Bool}\AgdaSpace{}%
\AgdaSymbol{→}\AgdaSpace{}%
\AgdaSymbol{(}\AgdaBound{x}\AgdaSpace{}%
\AgdaBound{y}\AgdaSpace{}%
\AgdaSymbol{:}\AgdaSpace{}%
\AgdaDatatype{Term}\AgdaSpace{}%
\AgdaGeneralizable{A}\AgdaSymbol{)}\AgdaSpace{}%
\AgdaSymbol{→}\AgdaSpace{}%
\AgdaDatatype{Term}\AgdaSpace{}%
\AgdaGeneralizable{A}\<%
\end{code}
\end{center}
\lessSpace{-0.8}

Then, we define its \AgdaFunction{normalize} function,
which takes an instance of \AgdaDatatype{Term}~$A$ and return an instance of type $A$.
The type guarantees that there will never be ill-typed terms like
\AgdaInductiveConstructor{succ}~(\AgdaInductiveConstructor{bool}~$x$).

\lessSpace{-0.8}
\begin{center}
\begin{code}%
\>[0]\AgdaFunction{normalize}\AgdaSpace{}%
\AgdaSymbol{:}\AgdaSpace{}%
\AgdaDatatype{Term}\AgdaSpace{}%
\AgdaGeneralizable{A}\AgdaSpace{}%
\AgdaSymbol{→}\AgdaSpace{}%
\AgdaGeneralizable{A}\<%
\\
\>[0]\AgdaFunction{normalize}\AgdaSpace{}%
\AgdaSymbol{(}\AgdaInductiveConstructor{nat}\AgdaSpace{}%
\AgdaBound{x}\AgdaSymbol{)}%
\>[24]\AgdaSymbol{=}\AgdaSpace{}%
\AgdaBound{x}\<%
\\
\>[0]\AgdaFunction{normalize}\AgdaSpace{}%
\AgdaSymbol{(}\AgdaInductiveConstructor{bool}\AgdaSpace{}%
\AgdaBound{x}\AgdaSymbol{)}%
\>[24]\AgdaSymbol{=}\AgdaSpace{}%
\AgdaBound{x}\<%
\\
\>[0]\AgdaFunction{normalize}\AgdaSpace{}%
\AgdaSymbol{(}\AgdaInductiveConstructor{succ}\AgdaSpace{}%
\AgdaBound{x}\AgdaSymbol{)}%
\>[24]\AgdaSymbol{=}\AgdaSpace{}%
\AgdaInductiveConstructor{suc}\AgdaSpace{}%
\AgdaSymbol{(}\AgdaFunction{normalize}\AgdaSpace{}%
\AgdaBound{x}\AgdaSymbol{)}\<%
\\
\>[0]\AgdaFunction{normalize}\AgdaSpace{}%
\AgdaSymbol{(}\AgdaInductiveConstructor{inv}\AgdaSpace{}%
\AgdaBound{t}\AgdaSymbol{)}%
\>[24]\AgdaSymbol{=}\AgdaSpace{}%
\AgdaFunction{not}\AgdaSpace{}%
\AgdaSymbol{(}\AgdaFunction{normalize}\AgdaSpace{}%
\AgdaBound{t}\AgdaSymbol{)}\<%
\\
\>[0]\AgdaFunction{normalize}\AgdaSpace{}%
\AgdaSymbol{(}\AgdaInductiveConstructor{case}\AgdaSpace{}%
\AgdaBound{b}\AgdaSpace{}%
\AgdaBound{x}\AgdaSpace{}%
\AgdaBound{y}\AgdaSymbol{)}%
\>[24]\AgdaSymbol{=}\AgdaSpace{}%
\AgdaOperator{\AgdaFunction{if}}\AgdaSpace{}%
\AgdaSymbol{(}\AgdaFunction{normalize}\AgdaSpace{}%
\AgdaBound{b}\AgdaSymbol{)}\<%
\\
\>[0][@{}l@{\AgdaIndent{0}}]%
\>[4]\AgdaOperator{\AgdaFunction{then}}\AgdaSpace{}%
\AgdaSymbol{(}\AgdaFunction{normalize}\AgdaSpace{}%
\AgdaBound{x}\AgdaSymbol{)}\AgdaSpace{}%
\AgdaOperator{\AgdaFunction{else}}\AgdaSpace{}%
\AgdaSymbol{(}\AgdaFunction{normalize}\AgdaSpace{}%
\AgdaBound{y}\AgdaSymbol{)}\<%
\end{code}
\end{center}
\lessSpace{-0.8}

We cannot encode \AgdaDatatype{Term} as a simpler indexed type because we cannot
pattern match on types, so the direct translation will not work --
We will need an auxiliary type to help us encoding them:

\lessSpace{-1.2}
\begin{align*}
&\kwdata~\Dat{TermTy}:\UU~\mid~\cons{natT}~\mid~\cons{boolT} \\
&\kw{func}~\func{termTy}~(x:\Dat{TermTy}):\UU \\[-0.3em]
& \mid~\cons{natT}\Rightarrow \Nat \\[-0.3em]
& \mid~\cons{boolT}\Rightarrow \Dat{Bool}
\end{align*}

Then, we define the type for terms and the normalize function:

\lessSpace{-1.2}
\begin{align*}
&\kwdata~\Dat{Term}~(n:\Dat{TermTy}):\UU \\[-0.3em]
& \mid~\cons{natT}~\Rightarrow\cons{nat}~\Nat \\[-0.3em]
& \mid~\cons{natT}~\Rightarrow\cons{succ}~(\Dat{Term}~\cons{natT}) \\[-0.3em]
& \mid~\cons{boolT}~\Rightarrow\cons{bool}~\cons{boolT} \\[-0.3em]
& \mid~\cons{boolT}~\Rightarrow\cons{inv}~(\Dat{Term}~\cons{boolT}) \\[-0.3em]
& \mid~A~\Rightarrow\cons{case}~(\Dat{Term}~\cons{boolT})~(\Dat{Term}~A)~(\Dat{Term}~A) \\
&\kw{func}~\func{normalize}~(t:\Dat{TermTy})~(x:\Dat{Term}~t):\func{termTy}~t \\[-0.3em]
& \mid~\cons{natT},\cons{nat}~n\Rightarrow n \\[-0.3em]
& \mid~\cons{natT},\cons{succ}~n\Rightarrow \succon~(\func{normalize}~\cons{natT}~n) \\[-0.3em]
& \mid~\cons{boolT},\cons{bool}~b\Rightarrow b \\[-0.3em]
& \mid~\cons{boolT},\cons{inv}~b\Rightarrow \func{not}~(\func{normalize}~\cons{boolT}~b) \\[-0.3em]
& \mid~t, \cons{case}~b~x~y\Rightarrow \func{ifElse}~(\func{normalize}~\cons{boolT}~b) \\[-0.3em]
& \hspace*{1em}~(\func{normalize}~t~x)~(\func{normalize}~t~y)
\end{align*}

In the general case, only when the indices are in canonical constructor form (say,
generated by references to parameters of the constructor and applications to constructors)
can we translate them into simpler index types.
Even though we could use auxiliary types to help us encoding them, there is still one
case where this encoding will fail, where the indices contain references
to the \textit{parameters} of the indexed type. The simplest case is the identity type:

\begin{code}[hide]%
\>[0]\AgdaSymbol{\{-\#}\AgdaSpace{}%
\AgdaKeyword{OPTIONS}\AgdaSpace{}%
\AgdaPragma{--cubical}\AgdaSpace{}%
\AgdaSymbol{\#-\}}\<%
\\
\>[0]\AgdaKeyword{open}\AgdaSpace{}%
\AgdaKeyword{import}\AgdaSpace{}%
\AgdaModule{Cubical.Core.Everything}\<%
\\
\>[0]\AgdaKeyword{open}\AgdaSpace{}%
\AgdaKeyword{import}\AgdaSpace{}%
\AgdaModule{Agda.Primitive}\<%
\\
\>[0]\AgdaKeyword{variable}\AgdaSpace{}%
\AgdaGeneralizable{ℓ}\AgdaSpace{}%
\AgdaSymbol{:}\AgdaSpace{}%
\AgdaPostulate{Level}\<%
\end{code}

\lessSpace{-0.8}
\begin{center}
\begin{code}%
\>[0]\AgdaKeyword{data}\AgdaSpace{}%
\AgdaDatatype{Id}\AgdaSpace{}%
\AgdaSymbol{(}\AgdaBound{A}\AgdaSpace{}%
\AgdaSymbol{:}\AgdaSpace{}%
\AgdaPrimitive{Type}\AgdaSpace{}%
\AgdaGeneralizable{ℓ}\AgdaSymbol{)}\AgdaSpace{}%
\AgdaSymbol{(}\AgdaBound{x}\AgdaSpace{}%
\AgdaSymbol{:}\AgdaSpace{}%
\AgdaBound{A}\AgdaSymbol{)}\AgdaSpace{}%
\AgdaSymbol{:}\AgdaSpace{}%
\AgdaBound{A}\AgdaSpace{}%
\AgdaSymbol{→}\AgdaSpace{}%
\AgdaPrimitive{Type}\AgdaSpace{}%
\AgdaBound{ℓ}\AgdaSpace{}%
\AgdaKeyword{where}\<%
\\
\>[0][@{}l@{\AgdaIndent{0}}]%
\>[2]\AgdaInductiveConstructor{idp}\AgdaSpace{}%
\AgdaSymbol{:}\AgdaSpace{}%
\AgdaDatatype{Id}\AgdaSpace{}%
\AgdaBound{A}\AgdaSpace{}%
\AgdaBound{x}\AgdaSpace{}%
\AgdaBound{x}\<%
\end{code}
\end{center}
\lessSpace{-0.8}

The index being $x$, a reference to the parameter of \AgdaDatatype{Id},
is the essential reason why a general term unification needs to be performed
during the pattern matching over \AgdaInductiveConstructor{idp}.
Pattern matching is a mechanism to match terms by patterns, not by terms.

Simpler indexed type essentially simplifies the problem of constructor selection just
by turning the term-match-term problem to a term-match-pattern problem,
which rules out numerous complication but also loses the benefit of general indexed types.
A potential way to bring general indexed types back without introducing them directly is discussed
as future work in~\cref{sub:future}, requiring the presence of a built-in identity type.


\subsection{Translation to indexed types} 
\label{sub:conserv}

We could translate simpler indexed types back to general indexed types.
To describe the translation, we define the syntax of general indexed types,
which is the output of the translation, in~\cref{fig:git}.
We do not have type parameters as they are just special cases of indices.

\begin{figure}
\begin{align*}
  decl' ::= & \quad \kwdata~\Dat D~\Delta~\overline{cons'} && \text{indexed type} \\[-0.3em]
  cons' ::= & \quad \mid \cons c:A && \text{constructor}
\end{align*}
\caption{Syntax of general indexed types}
\label{fig:git}
\end{figure}

Now we can start the translation.
First, we unify pattern matching constructors with simple constructors.
For constructor $\cons c~\Delta_{\cons c}$ in simpler indexed type $\Dat D~\Delta$, we translate
it into a pattern matching constructor $\vars\Delta \Rightarrow \cons c~\Delta_{\cons c}$.

After that, we could perform the translation of constructors.
In other words, we need to construct the type (``$A$'' in~\cref{fig:git}) of the translated constructor.

\begin{defn}
\label{defn:trans}
For pattern matching constructor $\overline p \Rightarrow \cons c~\Delta_{\cons c}$,
we construct the type of the translated constructor as
$\vars{\overline p}~\Delta_{\cons c} \to \Dat D~\toTerm{\overline p}$.
\end{defn}

This type is a pi type consisting of the following major components:

\begin{enumerate}
\item $\vars{\overline p}$: the bindings in the patterns. We turn these bindings into
 parameters of the translated type.
\item $\Delta_{\cons c}$: the constructor parameters. They are typed under the
 bindings in $\vars{\overline p}$, so we append the original parameters to the tail
 of these required bindings.
\item $\Dat D~\toTerm{\overline p}$: the return type. We specialize the indices of
 $\Dat D$ with the terms correspond to $\overline p$. These terms are typed under
 $\vars{\overline p}$, which is available in the domain of this pi type.
\end{enumerate}

\begin{namedthm}[Completeness]
Every simpler indexed type can be translated into general indexed types.
\end{namedthm}

\begin{proof}
This translation is defined for all simpler indexed types,
so we get the completeness theorem for free.
\end{proof}

\begin{namedthm}[Well-typedness]
The type of the translated constructor is well-scoped and well-typed.
\end{namedthm}
\begin{proof}
First, it indeed returns a specialization of the type $\Dat D$.

According to~\cref{fig:matchy}, types in $\Delta_{\cons c}$ are well-typed with
references to the bindings in $\overline p$, but not in $\Delta$.
These bindings are available in $\vars{\overline p}$.

According to~\cref{fig:terms}, the well-typedness of $\Dat D~\toTerm{\overline p}$
requires $\toTerm{\overline p} : \Delta$, and we know it is true by~\cref{lem:typed-pats}.
\end{proof}

\begin{namedthm}[Soundness]
The translated constructor needs to be matched if and only if the original constructor
needs to be matched.

The translated constructor does not need to be matched if and only if the original
constructor does not need to be matched.

The translated constructor cannot be matched if and only if the original constructor
cannot be matched.
\end{namedthm}

\begin{proof}
This theorem actually requires a bit more information to be well-defined -- we have not given
the general indexed types typing rules and semantics.

However, since the result of $\toTerm{\overline p}$ is only generated by applications to
constructors and references (by definition), we only need to deal with the unification of these terms,
which are quite simple. They should be structurally equivalent to the rules in~\cref{fig:patty}.

Restricting the unification problem to this smaller subset makes the soundness theorem provable
by induction on the patterns $\overline p$.
\end{proof}

\begin{remark}
This translation is also useful in the type checking of the constructors of simpler indexed types.
When we have a reference to such constructor without any argument supplied, we could synthesize
a type for this reference -- and we use the type in~\cref{defn:trans}.
\end{remark}


\subsection{Compilation and erasure} 
\label{sub:erasure}

The compilation of simpler indexed types has an advantage over normal indexed types.
In~\cite{NoIx}, they used detagging and forcing optimizations to erase the indices during compilation.
These methods are directly expressible in our syntax as the
indices are not even quantified in the constructors (see the examples in~\cref{sec:intro},
where the implicit argument $n$ in the Agda version of
\AgdaInductiveConstructor{fzero}, \AgdaInductiveConstructor{fsuc}
are not present in the corresponding definition as simpler indexed types).
In other words, simpler indexed types enjoy the benefit of index erasure
without any nontrivial compilation technique.

One downside is that there will not be inaccessible patterns~\cite{NoIx},
so there will be redundant pattern matching happening at runtime.
Consider the example in~\cref{sub:lim},
the $\func{normalize}$ function using simpler indexed type has two pattern matchings,
while the \AgdaFunction{normalize} function using general indexed types has only one.

In conclusion, simpler indexed types are more memory efficient than general indexed types
without optimizations, while they require redundant pattern matchings.
We think of the latter as a potential future work.


\section{Conclusion}
\label{sec:concl}
We introduced a simpler encoding of indexed types in dependent type systems.
It reuses the pattern matching for constructor selection to avoid exposing the index
unification problem to the users.
A number of existing indexed types such as \AgdaDatatype{Fin} and \AgdaDatatype{Vect} can
be encoded in this simpler way, but not all (exceptions include the identity type in Agda~\cite{Agda}
and the examples in~\cref{sub:motive}).

\subsection{Future work}
\label{sub:future}
We could translate simpler indexed types into an even simpler type theory
with only products and coproducts, just like in~\cite{Vec}.
This translation requires an algorithm to classify the pattern matching clauses
with overlapping parts. This is done in~\cite{OOP},
but in Aya we have a better implementation. We decide to describe such
translation after the overlapping pattern classification algorithm is formalized.

We could also have a built-in identity type in the type theory and encode the
indexed types with the identity type. The \AgdaDatatype{Image} type in~\cite{Vec} is a great example:

\begin{code}[hide]%
\>[0]\AgdaSymbol{\{-\#}\AgdaSpace{}%
\AgdaKeyword{OPTIONS}\AgdaSpace{}%
\AgdaPragma{--cubical}\AgdaSpace{}%
\AgdaSymbol{\#-\}}\<%
\\
\>[0]\AgdaKeyword{open}\AgdaSpace{}%
\AgdaKeyword{import}\AgdaSpace{}%
\AgdaModule{Cubical.Core.Everything}\<%
\\
\>[0]\AgdaKeyword{open}\AgdaSpace{}%
\AgdaKeyword{import}\AgdaSpace{}%
\AgdaModule{Agda.Primitive}\<%
\\
\>[0]\AgdaKeyword{variable}\AgdaSpace{}%
\AgdaGeneralizable{A}\AgdaSpace{}%
\AgdaGeneralizable{B}\AgdaSpace{}%
\AgdaSymbol{:}\AgdaSpace{}%
\AgdaPrimitive{Set}\<%
\\
\>[0]\AgdaKeyword{variable}\AgdaSpace{}%
\AgdaGeneralizable{ℓ}\AgdaSpace{}%
\AgdaSymbol{:}\AgdaSpace{}%
\AgdaPostulate{Level}\<%
\end{code}

\begin{center}
\begin{code}%
\>[0]\AgdaKeyword{data}\AgdaSpace{}%
\AgdaDatatype{Image}\AgdaSpace{}%
\AgdaSymbol{(}\AgdaBound{A}\AgdaSpace{}%
\AgdaBound{B}\AgdaSpace{}%
\AgdaSymbol{:}\AgdaSpace{}%
\AgdaPrimitive{Type}\AgdaSpace{}%
\AgdaGeneralizable{ℓ}\AgdaSymbol{)}\AgdaSpace{}%
\AgdaSymbol{(}\AgdaBound{f}\AgdaSpace{}%
\AgdaSymbol{:}\AgdaSpace{}%
\AgdaBound{A}\AgdaSpace{}%
\AgdaSymbol{→}\AgdaSpace{}%
\AgdaBound{B}\AgdaSymbol{)}\AgdaSpace{}%
\AgdaSymbol{:}\AgdaSpace{}%
\AgdaBound{B}\AgdaSpace{}%
\AgdaSymbol{→}\AgdaSpace{}%
\AgdaPrimitive{Type}\AgdaSpace{}%
\AgdaBound{ℓ}\AgdaSpace{}%
\AgdaKeyword{where}\<%
\\
\>[0][@{}l@{\AgdaIndent{0}}]%
\>[2]\AgdaInductiveConstructor{image}\AgdaSpace{}%
\AgdaSymbol{:}\AgdaSpace{}%
\AgdaSymbol{∀}\AgdaSpace{}%
\AgdaBound{x}\AgdaSpace{}%
\AgdaSymbol{→}\AgdaSpace{}%
\AgdaDatatype{Image}\AgdaSpace{}%
\AgdaBound{A}\AgdaSpace{}%
\AgdaBound{B}\AgdaSpace{}%
\AgdaBound{f}\AgdaSpace{}%
\AgdaSymbol{(}\AgdaBound{f}\AgdaSpace{}%
\AgdaBound{x}\AgdaSymbol{)}\<%
\end{code}
\end{center}
\lessSpace{-0.8}

It can be encoded as an inductive type without indices:

\lessSpace{-1.2}
\begin{align*}
&\kwdata~\Dat{Image}~(A~B:\UU)~(f:A \to B)~(b:B):\UU \\[-0.3em]
& \mid~\cons{image}~(x:A)~(p:f~x=b)
\end{align*}

We might be able to define a translation from general indexed types into simpler indexed types
with a built-in identity type, and during pattern matching over the encoded types, we perform
a rewriting over the identity proofs that we used to encode the indices.
By that, we will have a different treatment of the index unification problem,
and we could study how it compares to the general indexed types.

This idea (encoding the unification of type indices as a rewriting performed during pattern matching)
is similar to the \AgdaInductiveConstructor{transpX} operation discussed in~\cite[\S 3.2.4, \S 4]{TR-X}
and the ``index-fixing'' \textsf{fcoe} operation discussed in~\cite[\S 4.2]{HIT-CCTT},
but we are working in a general type theory with any definition of the identity type
as long as they support the J operation, including the path type in homotopy type theory~\cite{hottbook},
the path type in cubical type theories~\cite{CCHM, HIT-CCTT, CubicalAgda}, the identity type in
intuitionistic type theory~\cite{MLTT} (either homogeneous or heterogeneous), and others.

The cubical path type is a preferred choice as it does not depend on any fancy unification mechanism.
This means we can develop a type theory expressive enough to discuss indexed types
without dependent pattern matching.

Apart from that, we could seek integration with induction-recursion~\cite{IndRec}
and induction-induction~\cite{IndInd} as mentioned in~\cref{sub:decl}.

The compilation technique could be investigated to address the limitation discussed in~\cref{sub:erasure}.

\subsection{Related work}
\label{sub:related}
Type families in dependent types can be regarded as an encoding of GADTs~\cite{ALF}.
This idea was then put into a simpler type system (H-M) in~\cite{Silly},
and was developed further as \textit{first-class phantom types} in~\cite{GADT1}
and \textit{guarded recursive type constructors} in~\cite{GADT2}.
\cite{GADTMP} used Leibniz-style encoding of equality to reason over the equality among types
for building well-typed and well-scoped syntax trees.
GADTs are integrated into GHC Haskell in~\cite{GADTHS}.

Indexed types~\cite{IxTy} are the generalization of inductive types with type-equality,
where values are also allowed to appear as parameters of inductive types.
Agda~\cite{Agda} and Idris~\cite{Idris} have a more ergonomic design of indexed types
where the equality relations are made implicit.

In~\cite{Vec}, the type-family encoding of the sized vector type is discussed
and is directly related to simpler indexed types. However, there are several key advantages of
simpler indexed types over the record encoding given in~\cite{Vec}:

\begin{itemize}
\item Simpler indexed types have names for the types and constructors.
 The record encoding anonymizes the type and the constructors,
 so the error messages are harder to understand.
\item The pattern matching in simpler indexed type does not need to be covering.
 For instance, the simpler indexed type $\Dat{Fin}~\cons{zero}$ is implicitly an empty
 type, while encoding it as a function requires writing an explicit pattern matching
 clause $\AgdaInductiveConstructor{zero}=\AgdaDatatype{$\bot$}$.
\item Similar to coverage, pattern matching in simpler indexed types does not need
 to be structurally recursive. The record encoding uses functions so we need to
 respect the rules for functions, including persuading the termination checker.
\end{itemize}

Another work related to the encoding of indexed types is~\cite[\S 5]{GAL},
where they propose an encoding similar to the $\Dat{Image}$ example proposed in~\cref{sub:future}
and discuss a potential optimization to indexed types similar to~\cite{Vec}.
The advantages of simpler indexed types over~\cite{Vec} still applies to the encoding in~\cite{GAL}.
A notable application of indexed types based on~\cite{GAL} is \textit{ornaments}~\cite{Ornaments1, Ornaments2}.

The proposed feature has been implemented in two systems individually:
\begin{itemize}
\item
The Arend~\cite{Arend} proof assistant, an implementation of homotopy type theory with a cubical-flavored interval type.
\item
The Aya~\cite{Aya} proof assistant, an experimental implementation of a type theory similar to Arend's, but with
other features such as overlapping and order-independent patterns~\cite{OOP}.
\end{itemize}

All of the operations (except $\vars\Delta$ -- it is too simple to be a class) in~\cref{sub:ops} have a corresponding
class in the package \texttt{org.aya.core.pat} in the source code of Aya: $\vars p$ corresponds
to \texttt{PatTyper}, $\toTerm p$ corresponds to \texttt{PatToTerm}, $\matches u p$ corresponds to \texttt{PatMatcher}.

Apart from that, the type checking of terms in~\cref{fig:terms} corresponds to \texttt{ExprTycker},
the type checking of patterns in~\cref{fig:patty} corresponds to \texttt{PatTycker},
and the type checking of declarations in~\cref{fig:sig} corresponds to \texttt{StmtTycker}.
The source code of Aya could be retrieved from the link in the corresponding reference entry.
The complete normalizer example in~\cref{sub:lim} is available at
\url{https://github.com/aya-prover/aya-dev/blob/main/base/src/test/resources/success/type-safe-norm.aya}
as a test-case of the Aya type checker.

\subsection{Acknowledgments}
We are grateful to Yu Zhang and Yiming Zhu for their suggestions on the draft version of this paper
and Guillaume Allais for their review on an earlier version of this paper.
We would also like to thank the anonymous reviewers from the ICFP TyDe workshop
for their valuable suggestions on improving the paper.
\printbibliography
\end{document}